\documentclass[%
reprint,
 amsmath,amssymb,
 aps,
 pra,
 superscriptaddress
]{revtex4-1}
\usepackage{breakurl}
\usepackage{graphicx}
\usepackage{dcolumn}
\usepackage{bm}
\usepackage{hyperref}
\hypersetup{hidelinks,
            colorlinks=true,
            allcolors=blue}
\usepackage{amsthm}
\usepackage{yquant}
\useyquantlanguage{groups}
\usepackage{tikz}
\tikzset{
  dot/.style={
    circle, fill=black, inner sep=1pt, outer sep=0pt
  }
}
\usetikzlibrary{fit}

\usepackage[ruled,linesnumbered,noline]{algorithm2e}
\usepackage{diagbox}
\usepackage{pgfplots}    
\pgfplotsset{width=12cm,compat=1.13}
\usepackage{cases}
\newtheorem{definition}{Definition}

\newtheorem{lemma}{Lemma}

\usepackage{float}
\usepackage{subfigure}
\usepackage{placeins}
\usepackage{todonotes}
\usepackage{thm-restate}

\begin{document}

\preprint{APS/123-QED}

\title{Nearly Optimal T-Count for Symmetry-Based Quantum State Purification}
\author{Jialiang Tang}
\thanks{These authors contributed equally to this work.}
 \affiliation{%
 State Key Lab of Processors, Institute of Computing Technology, Chinese Academy of Sciences, Beijing 100190, China 
}%
\affiliation{%
 School of Computer Science and Technology, University of Chinese Academy of Sciences, Beijing 100049, China
}%
 \author{Junhong Nie}%
\thanks{These authors contributed equally to this work.}
\affiliation{Shandong University, Jinan 250100, China}

\author{Keming He}%
\affiliation{%
Thrust of Artificial Intelligence, Information Hub,
The Hong Kong University of Science and Technology (Guangzhou), Guangdong 511453, China.
}%

 \author{Xin Wang}%
  \email{felixxinwang@hkust-gz.edu.cn}
\affiliation{%
Thrust of Artificial Intelligence, Information Hub,
The Hong Kong University of Science and Technology (Guangzhou), Guangdong 511453, China.
}%
 \author{Xiaoming Sun}%
 \email{sunxiaoming@ict.ac.cn}
\affiliation{%
 State Key Lab of Processors, Institute of Computing Technology, Chinese Academy of Sciences, Beijing 100190, China 
}%
\affiliation{%
 School of Computer Science and Technology, University of Chinese Academy of Sciences, Beijing 100049, China
}%


\begin{abstract}
Protecting quantum information from noise is essential for reliable quantum computation and quantum memories. Quantum state purification addresses this challenge by processing multiple noisy copies of an unknown quantum state to produce an output state of higher fidelity. Here we establish nearly tight upper and lower bounds on the non-Clifford resources required to implement the symmetric-projection purification protocol. For \(n\) copies of a \(d\)-dimensional system, we construct a Clifford+T circuit that block-encodes the symmetric projector with operator-norm error at most \(\epsilon\) using \(\widetilde{O}(n)\) T gates, where the tilde-$O$ notation suppresses polylogarithmic factors in \(n\), \(d\), and \(1/\epsilon\). Conversely, using stabilizer nullity, we prove that any Clifford+T implementation of such a block encoding with error \(\epsilon<\frac{1}{2}\sqrt{1-1/n}\) requires at least \(2n-2\) T gates, even allowing arbitrary positive normalization and clean stabilizer ancillas. Taken together, these results establish the nearly optimal scaling of non-Clifford cost with the number of copies and uncover the fundamental resource requirements for coherent symmetry projection in quantum state purification.
\end{abstract}

\maketitle

\section{Introduction}\label{intro}
Quantum computing is currently transitioning from the Noisy Intermediate-Scale Quantum (NISQ) era to fault-tolerant quantum computing. However, noise, decoherence, and control imperfections remain ubiquitous in physical quantum devices, severely limiting the development of quantum computation and quantum information processing. To address these challenges, a variety of promising techniques have been proposed. One such approach is quantum error correction \cite{Shor1995,Steane1996,Fowler2012,Breuckmann2021,Campbell2017,Acharya2025,Bluvstein2024}, which protects quantum information by redundantly encoding logical states into enlarged physical Hilbert spaces. In parallel, quantum state purification for the stabilization of quantum computations, including quantum state storage, has emerged as an important complementary technique. This approach has also recently been referred to as quantum purity amplification in the literature.\cite{Barenco1997,Cirac1999,KeylWerner2001,Fiurasek2004,Fu2016,Fang2020NoGo,Liu2026NoUniversal,He2026NoGo,Childs2025,Yao2025,Li2025QPA,Zhao2026Distributed,He2026Magic,Guo2026Energy,He2026b}. 

Quantum state purification collectively processes multiple noisy copies of an unknown quantum state to produce fewer output states with higher fidelity \cite{Childs2025,Yao2025,Grier2025,Li2025QPA}, with important applications in the fields of quantum information and quantum communication. The most thoroughly researched case is $n$-to-1 purification, in which noisy states $\rho_{\psi}^{\otimes n} = \mathcal{N}(\psi^{\otimes n})$ under noisy channel $\mathcal{N}$ is turned to one purified output $\sigma_{\psi}$.

For $n$-to-1 state purification, one line of research aims to optimize the averaged fidelity of outputs. Cirac, Ekert and Macchiavello \cite{Cirac1999} gave the first solution for purifying identically prepared noisy states with optimal averaged fidelity, known as  Cirac-Ekert-Macchiavello (CEM) protocol. This work made explicit the common representation theory structure behind purification, mixed-state estimation, and cloning. Following that, Keyl and Werner \cite{KeylWerner2001} constructed optimal purification operations in a more general setting and characterized the attainable output rates. They showed that the rate depends crucially on the fidelity criterion: high fidelity of each individual output is substantially weaker than high global fidelity with an uncorrelated tensor power of the target. This distinction between one-site and all-site guarantees has remained important in later multi-output formulations \cite{Scharnhorst2026}. Also, subsequent experimental studies investigated its implementation on different physical platforms. Ricci et al. \cite{Ricci2004} first demonstrated a two-copy implementation for depolarized qubits using passive linear-optical elements. Hou et al. \cite{Hou2014} later realized the CEM protocol with explicit symmetrization in a nuclear-magnetic-resonance quantum information processor.

The other line studies probabilistic purification, which optimizes the conditional output fidelity together with the probability of obtaining the successful outcome \cite{Fiurasek2004,Fu2016,Childs2025,Yao2025,Li2025QPA}. In this setting, the purified state is accepted only when the purification succeeds, and failed cases are discarded. Fiur\'{a}\v{s}ek \cite{Fiurasek2004} formulated probabilistic purification as optimization over completely positive trace-non-increasing maps and derived the optimal probabilistic purification protocol. Yao et al. \cite{Yao2025} developed the symmetric projector protocol by block-encoding,
\begin{equation*}
    U_{n,d}
    =\begin{pmatrix}
        \Pi_{n,d} & *\\
        * & *
      \end{pmatrix}.
\end{equation*}
where
\begin{equation*}
    \Pi_{n,d}
    :=\frac{1}{n!}\sum_{\pi\in S_n}P_n(\pi)
\end{equation*}
is the uniform superposition of permutations on qudits with dimension $d$. Under depolarizing noise, it is proved to be jointly optimal in the trade-off of fidelity and probability for 2-to-1 state purification with arbitrary dimension, with numerical optimality in other cases \cite{Cirac1999, Fu2016,Yao2025}. Such protocol projects noisy states into the symmetric space to raise the fidelity, originating from SWAP test in early research \cite{Barenco1997, Buhrman2001}, which has been adopted broadly in state purification \cite{Fu2016, Childs2025,Grier2025}.

The scope of purification then has been expanded beyond $n$-to-1 settings. Li et al. \cite{Li2026Multiple} characterized the optimal channel for $n$-to-$m$ purification. Independently, Scharnhorst et al. \cite{Scharnhorst2026} proved dimension-independent sample bounds for $n$-to-$m$ purification using the combinatorics of random Young diagrams.

Despite these studies on the performance of purification protocols, the remaining question is how efficiently they can be implemented fault tolerantly. This efficiency can be measured by T-count (i.e. the number of T gates) \cite{Gottesman1997,GottesmanChuang1999,BravyiKitaev2005,JiangWang2023Nullity,Li2023}. In this work, we consider the $n$-to-1 probabilistic purification problem and focus on the fixed symmetric projector protocol $U_{n,d}$ due to its optimality under depolarizing noise. The synthesis objective is to approximate its successful Kraus operator $\Pi_{n,d}$ in operator norm while minimizing T-count.
On this problem, a natural mathematical description for such operator is Schur--Weyl duality, with the Quantum Schur Transform (QST) serving as the corresponding circuit primitive~\cite{Harrow2005Schur,Harrow2013,Bacon2006,watts2026}. After standard Clifford+T compilation, the BCH algorithm gives quantum circuits for QST with T-count polynomial in $n$, $d$, and $\log(1/\varepsilon)$~\cite{Bacon2006}. Krovi later proposed a dual, high-dimensional construction intended to improve the dependence on $d$~\cite{Krovi2019}. A subsequent analysis gave a corrected Krovi algorithm with T-count $\widetilde{O}(n^4)$ when $d>n$, as well as a compressed BCH construction with T-count $\widetilde{O}(\min\{n^5,nd^4\})$~\cite{Burchardt2025}. Recently, Brahmachari et al. showed that random pairwise SWAP tests can approximate unitary Schur sampling and attain the optimal purification protocol using $\widetilde{O}(\min\{n^4,nd^3\})$ tests~\cite{Brahmachari2025,Brahmachari2026}. However, Schur-Weyl duality is more general than required to implement $U_{n,d}$, since it resolves all Schur--Weyl sectors rather than only the completely symmetric sector. This motivates a more compact circuit of the symmetric projector $U_{n,d}$ for $n,d\geq 2$.

In this work, we give a Clifford+T synthesis scheme of the symmetric projector protocol $U_{n,d}$. For $n\geq2$ input copies of local dimension $d\geq2$, our construction approximates the symmetric projector in operator norm with T-count $\widetilde O(n)$, where the notation suppresses logarithmic factors in $n$, $d$, and $1/\epsilon$. Figure~\ref{fig:clean-partition} illustrates the permutation circuit structure in our framework.

\begin{samepage}
\begin{restatable}{theorem}{upperbound}\label{thm:symmetric-projection-synthesis}
For $n\geq2$, $d\geq2$, and $0<\epsilon<1/2$, there is a Clifford+T unitary $\widetilde U_{n,d}$ that block-encodes $\Pi_{n,d}$ with normalization one and operator-norm error at most $\epsilon$, using $m=O(n\log n)$ auxiliary qubits and T-count $T(\widetilde U_{n,d})=\widetilde O(n)$.
\end{restatable}
\end{samepage}

\begin{table*}[t]
\caption{Comparison of T-count upper bounds.  Previously reported circuit and
SWAP-test complexities are expressed as T-counts under standard
Clifford+T compilation; poly-logarithmic overheads of $n$, $d$ and $\epsilon$ are absorbed
by $\widetilde O(\cdot)$. The rows implement different circuit primitives, as
indicated in the second column.}
\label{tab:t-count-comparison}
\begin{ruledtabular}
\begin{tabular}{lccc}
\textbf{Method} & \textbf{Implementation} & \textbf{T-count} & \textbf{Regime} \\
\hline
Corrected Krovi~\cite{Krovi2019,Burchardt2025}
& QST-based & $\widetilde O(n^4)$
& $d>n\geq 2$, $0<\epsilon<1$ \\
BCH / compressed BCH~\cite{Bacon2006,Burchardt2025}
& QST-based & $\widetilde O(\min\{n^5,nd^4\})$
& $n,d\geq2$, $0<\epsilon<1$ \\
Random SWAP~\cite{Brahmachari2025,Brahmachari2026}
& SWAP tests & $\widetilde{O}(\min\{n^4,nd^3\})$
& $n,d\geq2$, $0<\epsilon<1$ \\
\textbf{This work}
& Dicke balanced tree mask & $\widetilde O(n)$
& $n,d\geq2$, $0<\epsilon<1/2$ \\
\end{tabular}
\end{ruledtabular}
\end{table*}

As summarized in Table~\ref{tab:t-count-comparison}, compared with the best reported high-dimensional QST T-count bound $\widetilde{O}(n^4)$, our direct projector construction has nearly linear T-count and only logarithmic dependence on the local dimension $d$. The key technical ingredients are a compact control architecture for Linear Combination of Unitary (LCU) and a sorting network construction based on Dicke states, encoded to generate the required permutations. We also prove a robust lower bound for the same operator-synthesis task, using stabilizer nullity~\cite{Beverland2020Nullity,JiangWang2023Nullity}.

\begin{restatable}{theorem}{lowerbound}\label{thm:symmetric-projection-lower}
For $n\geq2$, $d\geq2$, and $0\leq\epsilon<\tfrac12\sqrt{1-1/n}$, let $U$ be any block encoding of $\Pi_{n,d}$ with operator-norm error at most $\epsilon$, and let $V$ be any Clifford+T implementation of $U$ using clean stabilizer work qubits. Then $T(V)\geq2n-2$. The bound allows every finite positive normalization and any number of auxiliary qubits.
\end{restatable}

Taking together Theorem \ref{thm:symmetric-projection-synthesis} and \ref{thm:symmetric-projection-lower}, our $\widetilde O(n)$ construction and the $2n-2$ lower bound establish a nearly optimal T-count for implementing the symmetric projector protocols, which has been demonstrated to attain the optimal purification performance under depolarizing noise. The comparison between previous research and our results is shown in Table \ref{tab:t-count-comparison}.

The rest of the paper is organized as below. Section \ref{sec:preliminaries}
is about the preliminaries. In section \ref{sec3}, we introduce our synthesis framework in detail, then provide a T-count analysis for approximating purification protocols. In section \ref{sec4} we prove the robust linear lower bound for block-encoding the symmetric projector. We make a summary of this paper in section \ref{sec5}, including achievements and drawbacks, and then raise some open problems for future work.

\section{Preliminaries}
\label{sec:preliminaries}

This section proceeds in three steps. Section~\ref{subsec:probabilistic-purification} defines probabilistic $n$-to-1 purification and its probability--fidelity trade-off. Section~\ref{subsec:depolarizing-noise} specifies the depolarizing-noise model, and Sec.~\ref{subsec:symmetric-projection-protocol} introduces the symmetric projector protocol and converts its implementation into the problem of T gate synthesis of $\Pi_{n,d}$.

\subsection{Probabilistic $n$-to-1 quantum state purification}
\label{subsec:probabilistic-purification}

Let $\mathcal{H}_A \cong \mathbb{C}^{d}$ be the Hilbert space of a
$d$-dimensional quantum system.  For any Hilbert space $\mathcal{H}$, we write
$\mathcal{L}(\mathcal{H})$ for the vector space of linear operators acting on
$\mathcal{H}$. We denote an unknown target pure state $\psi$ by its density matrix
$|\psi\rangle\!\langle\psi| \in \mathcal{L}(\mathcal H_A)$ and the corresponding noisy state by
$\rho_{\psi}:=\mathcal{N}(\psi)$, where $\mathcal{N}$ is a fixed noise
channel.  An $n$-to-1 probabilistic purification protocol takes $n$
identically prepared noisy copies $\rho_{\psi}^{\otimes n}$ as input and is
described, on its successful branch, by a completely positive and
trace-nonincreasing map
\begin{equation*}
    \mathcal{E}_{A_1\cdots A_n\to B}:
    \mathcal{L}\!\left(\mathcal{H}_A^{\otimes n}\right)
    \longrightarrow \mathcal{L}(\mathcal{H}_B),
    \qquad \mathcal{H}_B\cong\mathcal{H}_A.
\end{equation*}
Such a map is abbreviated as a CPTN map \cite{Choi1975,Jamiolkowski1972}.  Complete positivity ensures that the
map remains positive when acting jointly with an arbitrary reference system,
whereas the trace-nonincreasing condition accounts for the probability of the
selected measurement outcome.
For a given target state $\psi$, the unnormalized and normalized outputs are
\begin{equation*}
    \widetilde{\sigma}_{\psi}
    :=\mathcal{E}\!\left(\rho_{\psi}^{\otimes n}\right),
    \qquad
    \sigma_{\psi}
    :=\frac{\widetilde{\sigma}_{\psi}}{p_{\psi}},
\end{equation*}
where $p_{\psi}:=\operatorname{Tr}\!\left(\widetilde{\sigma}_{\psi}\right)$ is the probability that the protocol succeeds.  Conditioned
on success, its purification fidelity is
\begin{equation*}
    f_{\psi}
    :=\langle\psi|\sigma_{\psi}|\psi\rangle
    =\frac{\operatorname{Tr}\!\left[
        \psi\,\mathcal{E}(\rho_{\psi}^{\otimes n})
    \right]}{p_{\psi}}.
\end{equation*}
Thus, unlike deterministic purification, a probabilistic protocol may improve
the conditional output fidelity by rejecting unsuccessful outcomes.

For universal purification, the target state is unknown and the performance is
averaged over the Haar measure $\mathrm{d}\psi$ \cite{Harrow2013,Mele2024}.  The average success probability and
the corresponding conditional average fidelity of $\mathcal{E}$ are
\begin{align*}
    p(\mathcal{E})
    &:=\int \operatorname{Tr}\!\left[
        \mathcal{E}(\rho_{\psi}^{\otimes n})
    \right]\mathrm{d}\psi,
    \\
    f(\mathcal{E})
    &:=\frac{1}{p(\mathcal{E})}
      \int \operatorname{Tr}\!\left[
        \psi\,\mathcal{E}(\rho_{\psi}^{\otimes n})
      \right]\mathrm{d}\psi.
\end{align*}
The two complementary optimization problems are therefore
\begin{align*}
    F_{\mathcal{N}}(n,p)
    &:=\max_{\mathcal{E}\in\mathrm{CPTN}}
      \left\{f(\mathcal{E}):p(\mathcal{E})=p\right\},
    \\
    P_{\mathcal{N}}(n,f)
    &:=\max_{\mathcal{E}\in\mathrm{CPTN}}
      \left\{p(\mathcal{E}):f(\mathcal{E})=f\right\}.
\end{align*}
The first quantity is the largest fidelity attainable with a prescribed
success probability, whereas the second is the largest success probability
attainable at a prescribed fidelity.  Together, they characterize the
probability--fidelity trade-off of universal probabilistic purification.

\subsection{Depolarizing noise}
\label{subsec:depolarizing-noise}

We now specialize the preceding purification task to inputs affected by the $d$-dimensional depolarizing
channel \cite{Preskill2018,Urbanek2021}
\begin{equation*}
    \mathcal{D}_{\delta}(X)
    :=(1-\delta)X
      +\delta\,\operatorname{Tr}(X)\frac{I_d}{d},
    \qquad 0\leq\delta\leq 1.
\end{equation*}
For a pure target state $\psi$, the resulting noisy state is
\begin{equation*}
    \rho_{\psi}
    =\mathcal{D}_{\delta}(\psi)
    =\lambda_0\psi
      +\lambda_{\perp}(I_d-\psi),
\end{equation*}
where
\begin{equation*}
    \lambda_0
    =1-\frac{d-1}{d}\delta,
    \qquad
    \lambda_{\perp}
    =\frac{\delta}{d}.
\end{equation*}
Hence the unpurified input fidelity is
$\langle\psi|\rho_{\psi}|\psi\rangle=\lambda_0$.  In a basis whose first
vector is $|\psi\rangle$, the spectrum of $\rho_{\psi}$ is represented by
\begin{equation*}
    \Lambda
    :=\operatorname{diag}\!\left(
        \lambda_0,
        \underbrace{\lambda_{\perp},\ldots,\lambda_{\perp}}_{d-1}
      \right).
\end{equation*}
The channel is unitarily covariant, namely
\begin{equation*}
    \mathcal{D}_{\delta}(UXU^{\dagger})
    =U\mathcal{D}_{\delta}(X)U^{\dagger}
\end{equation*}
for every unitary $U$.  Consequently, any covariant purification protocol has
the same success probability and conditional fidelity for every target
$|\psi\rangle$; these quantities can therefore be evaluated using the fixed
spectrum $\Lambda$.

\subsection{Symmetric projector operator and block encoding}
\label{subsec:symmetric-projection-protocol}

We now specify the purification protocol whose circuit cost we analyze. Let $S_n$ be the symmetric group on $n$ elements and let $P_n(\pi)$ denote the
unitary that permutes the $n$ qudit copies of dimension $d$ according to $\pi\in S_n$.  The
projector onto the completely symmetric subspace of
$\mathcal{H}_A^{\otimes n}$ is
\begin{equation*}
    \Pi_{n,d}
    :=\frac{1}{n!}\sum_{\pi\in S_n}P_n(\pi).
\end{equation*}
For depolarized i.i.d. inputs, the purification protocol studied in Ref.~\cite{Yao2025} is the CPTN map
\begin{equation*}
    \mathcal{E}^{\star}_{A_1\cdots A_n\to A_1}(X)
    :=\operatorname{Tr}_{A_2\cdots A_n}\!\left[
        \Pi_{n,d} X\Pi_{n,d}
      \right].
\end{equation*}
Operationally, the protocol projects the $n$ noisy copies onto their completely
symmetric component, accepts only the successful projection outcome, and then
discards $A_2,\ldots,A_n$.  On input $\rho_{\psi}^{\otimes n}$, its success
probability and normalized output are
\begin{align*}
    p_n
    &:=\operatorname{Tr}\!\left[
        \Pi_{n,d}\rho_{\psi}^{\otimes n}
      \right],
    \\
    \sigma_{n,\psi}
    &:=\frac{1}{p_n}
      \operatorname{Tr}_{A_2\cdots A_n}\!\left[
        \Pi_{n,d}\rho_{\psi}^{\otimes n}\Pi_{n,d}
      \right],
\end{align*}
and the conditional fidelity is
\begin{equation*}
    f_n:=\langle\psi|\sigma_{n,\psi}|\psi\rangle.
\end{equation*}
The purification mechanism has a simple symmetry-based interpretation.  The
ideal component $\psi^{\otimes n}$ is invariant under every permutation of the
input copies and therefore lies entirely in the completely symmetric
subspace.  Noise, by contrast, produces components in which different copies
occupy different eigenstates; although the full state
$\rho_{\psi}^{\otimes n}$ is permutation invariant as a density operator,
these components generally have support in both symmetric and nonsymmetric
representation sectors.  Postselection with $\Pi_{n,d}$ removes all nonsymmetric
sectors without attenuating the ideal vector $|\psi\rangle^{\otimes n}$.
Consequently, conditioned on a successful projection, the relative weight of
the ideal component increases, which is the source of the fidelity
improvement.

Since $\rho_{\psi}^{\otimes n}$ is permutation invariant, it commutes with
$\Pi_{n,d}$.  Consequently, the second projector in the definition of
$\sigma_{n,\psi}$ may equivalently be omitted when computing the success
probability or expectation values.

For the depolarizing spectrum $\Lambda$ defined above, both performance
measures admit the recursive expressions
\begin{align*}
    p_0&:=1,
    &
    p_n
    &=\frac{1}{n}\sum_{j=1}^{n}
      p_{n-j}\operatorname{Tr}\!\left(\Lambda^j\right),
    \\
    &&
    f_n
    &=\frac{1}{np_n}\sum_{j=1}^{n}
      p_{n-j}\lambda_0^j.
\end{align*}
In particular, for two input copies,
\begin{equation*}
    p_2=\frac{1+\operatorname{Tr}(\Lambda^2)}{2},
    \qquad
    f_2=\frac{\lambda_0+\lambda_0^2}{2p_2}.
\end{equation*}
Yao et al. \cite{Yao2025} proved that $(p_2,f_2)$ is a jointly optimal ``golden
point'' for arbitrary local dimension $d$:
\begin{equation*}
    F_{\mathcal{D}_{\delta}}(2,p_2)=f_2,
    \qquad
    P_{\mathcal{D}_{\delta}}(2,f_2)=p_2.
\end{equation*}
Moreover,
$F_{\mathcal{D}_{\delta}}(2,p)=f_2$ for every $0<p\leq p_2$. It means
the success probability cannot be traded for a fidelity beyond
$f_2$. This establishes the optimal
probability--fidelity trade-off at the golden point. For general $n$, the same
symmetric projector protocol yields the pair $(p_n,f_n)$ specified by the
recursions above. Its optimality was numerically verified in
Ref.~\cite{Yao2025} for $n\leq 8$ and $d\leq 10$ and conjectured for arbitrary
$n$ and $d$. The analytic proof in their work applies to $n=2$.

It remains to describe how the nonunitary projector $\Pi_{n,d}$ is realized by a
quantum circuit. Given an operator $M$ on the input copies, a unitary $U$
on $m$ ancilla qubits and the system is an exact block encoding \cite{Gilyen2019QSVT,Camps2024,LowChuang2019} of $M$ if
\begin{equation*}
    (\langle 0^m|\otimes I)U(|0^m\rangle\otimes I)=M.
\end{equation*}
Equivalently, in the decomposition induced by the ancilla state $|0^m\rangle$,
the encoded operator appears in the upper-left block,
\begin{equation*}
    U=
    \begin{pmatrix}
        M & *\\
        * & *
    \end{pmatrix}.
\end{equation*}
Because the expression for $\Pi_{n,d}$ above is a uniform linear combination of
permutation unitaries, it admits a standard linear-combination-of-unitaries
(LCU) block encoding \cite{ChildsWiebe2012,LowChuang2019,Gilyen2019QSVT}.  Let $m=\lceil\log_2(n!)\rceil$, let the valid computational
basis states of the ancilla be labeled by permutations, and define
\begin{align*}
    P|0^m\rangle
    &:=\frac{1}{\sqrt{n!}}
      \sum_{\pi\in S_n}|\pi\rangle,
    \\
    S
    &:=\sum_{\pi\in S_n}
      |\pi\rangle\!\langle\pi|\otimes P_n(\pi),
\end{align*}
where the selection unitary $S$ is extended arbitrarily but unitarily on unused
ancilla basis states.  The unitary
\begin{equation*}
    U_{n,d}
    :=(P^{\dagger}\otimes I)\,
      S\,
      (P\otimes I)
\end{equation*}
satisfies
\begin{equation*}
    (\langle0^m|\otimes I)U_{n,d}
    (|0^m\rangle\otimes I)
    =\frac{1}{n!}\sum_{\pi\in S_n}P_n(\pi)
    =\Pi_{n,d}.
\end{equation*}
Thus $U_{n,d}$ has the form
\begin{equation*}
    U_{n,d}
    =\begin{pmatrix}
        \Pi_{n,d} & *\\
        * & *
      \end{pmatrix}.
\end{equation*}
Preparing the ancilla in $|0^m\rangle$, applying $U_{n,d}$, measuring the
ancilla, and postselecting the all-zero outcome applies the Kraus operator
$\Pi_{n,d}$ to the input.  The probability of this outcome is exactly $p_n$; after
tracing out $A_2,\ldots,A_n$, the remaining register is the purified state
$\sigma_{n,\psi}$.  Hence implementing the purification protocol does not
require a complete decomposition into all Schur--Weyl sectors: it suffices to
distinguish the completely symmetric sector from its orthogonal complement.
Thus, the purification performance of this fixed protocol is the pair
$(p_n,f_n)$, while the circuit-synthesis target is a unitary whose upper-left
block approximates its successful Kraus operator $\Pi_{n,d}$ in operator norm.

\section{Synthesis framework and resource estimation}\label{sec3}

Section~\ref{subsec:symmetric-projection-protocol} reduces the implementation
of the purification protocol to a block encoding of
$\Pi_{n,d}=n!^{-1}\sum_{\pi\in S_n}P_n(\pi)$. Standard LCU therefore requires a
uniform superposition of permutation instructions and a selection unitary that coherently applies $P_n(\pi)$. A direct conditional controlled design (e.g. the implementation in \cite{Yao2025}) for each permutation does not by itself provide an efficient uniform selection circuit. We address this issue by encoding every permutation by a uniformly-controlled circuit structure.

We construct a block encoding of $\Pi_{n,d}$ by preparing a uniform
superposition of permutation instructions, applying the instructed
permutation to the input copies, and reversing the preparation.  The
instructions consist of binary masks arranged on a balanced binary tree.
We first explain this encoding, then give its preparation and routing
circuits, and finally derive the LCU block and its approximation error.

Throughout this section, $D$ denotes the Dicke-mask register, $C$ the
control register holding the sorting history, and $A$ the system holding
the input copies.  We write $P$ for preparation and retain $P_n(\pi)$ for
the permutation of the copies; the subscript and argument distinguish
these two operators.  The total number of auxiliary qubits is denoted by
$m$, and the number of switches in a size-$k$ sorting network by $f(k)$.
We first consider qubit copies in the Clifford+T model with unrestricted
connectivity.  All logarithms are base two.  Sorting and routing are
exact; approximation enters only through Dicke-state preparation.

\subsection{Encoding permutations by recursive subsets}
\label{subsec:recursive-subsets}

We first explain how a collection of masks determines one complete
permutation. A \emph{mask} for a block of $k$ copies is a binary string
$x=x_1\cdots x_k$ specifying which copies will enter a left subblock.
A bit $x_i=1$ selects the copy currently at position $i$, whereas $x_i=0$
assigns it to the right subblock.  We always select
$w=\lfloor k/2\rfloor$ copies, so valid masks have Hamming weight
$|x|:=\sum_i x_i=w$.  For example, the mask $1010$ applied to the ordered
copies $(a,b,c,d)$ selects $a,c$ for the left subblock and $b,d$ for the
right subblock.  The symbols here identify positions for the explanation;
no classical description or identifying label of a quantum state is
stored by the circuit.

Place this selection rule on a balanced binary tree with $n$ leaves.
The root contains all $n$ copies.  An internal node $v$ containing $k_v$
copies has a left child of size $w_v=\lfloor k_v/2\rfloor$ and a right
child of size $k_v-w_v$; a leaf specifies one final output position.
Associate a separate mask $x_v$ with every internal node.  As illustrated
in Fig.~\ref{fig:mask-tree}, the root mask first determines the two groups,
and the child masks then specify how to split each resulting group.
Thus a child mask refers to positions \emph{after} its parent's
partition, rather than to the original input positions.

\begin{figure}[tbp]
\centering
\begin{tikzpicture}[x=1cm,y=1cm,
 every node/.style={font=\footnotesize},
 treeblock/.style={draw=blue!55!black,fill=blue!3,rounded corners=2pt,
                  align=center,inner sep=5pt},
 treeleaf/.style={draw=black!50,fill=black!3,minimum width=0.6cm,
                 minimum height=0.45cm},
 edge/.style={->,>=stealth,draw=black!65}]
\node[treeblock,text width=4.2cm] (root) at (4,4.45)
 {$k=4$, $w=2$: input $(a,b,c,d)$\\
  $D_4$ prepared in $|D_4^2\rangle$\\
  branch $x_4=1010$: $(a,c\mid b,d)$};
\node[treeblock,text width=2.8cm] (left) at (2,2.05)
 {input $(a,c)$\\
  $D_{2,L}$ in $|D_2^1\rangle$\\
  $x_{2,L}=01$: $(c\mid a)$};
\node[treeblock,text width=2.8cm] (right) at (6,2.05)
 {input $(b,d)$\\
  $D_{2,R}$ in $|D_2^1\rangle$\\
  $x_{2,R}=01$: $(d\mid b)$};
\draw[edge] (root.south)--node[pos=.55,above left,inner sep=1pt]
 {selected: $a,c$} (left.north);
\draw[edge] (root.south)--node[pos=.55,above right,inner sep=1pt]
 {remaining: $b,d$} (right.north);
\foreach \xx/\symbol/\name in {1/c/lone,3/a/ltwo,5/d/rone,7/b/rtwo}{
 \node[treeleaf] (\name) at (\xx,0.25) {$\symbol$};
}
\draw[edge] (left.south)--node[above left,inner sep=1pt] {$1$} (lone.north);
\draw[edge] (left.south)--node[above right,inner sep=1pt] {$0$} (ltwo.north);
\draw[edge] (right.south)--node[above left,inner sep=1pt] {$1$} (rone.north);
\draw[edge] (right.south)--node[above right,inner sep=1pt] {$0$} (rtwo.north);
\node[align=center] at (4,-0.55)
 {$z=x_4\Vert x_{2,L}\Vert x_{2,R}=1010\Vert01\Vert01$};
\end{tikzpicture}
\caption{Balanced binary tree for four copies.  Each internal node has
a Dicke mask that selects the copies sent to its left child.  The figure
follows one computational branch, $z=1010\Vert01\Vert01$, from input
$(a,b,c,d)$ to output $(c,a,d,b)$.  Child masks refer to the positions
after the parent partition.  All $6\cdot2\cdot2=24$ valid collections
are present with equal amplitude.  The root's intermediate ordering
uses the five-switch network displayed in Fig.~\ref{fig:clean-partition}.}
\label{fig:mask-tree}
\end{figure}
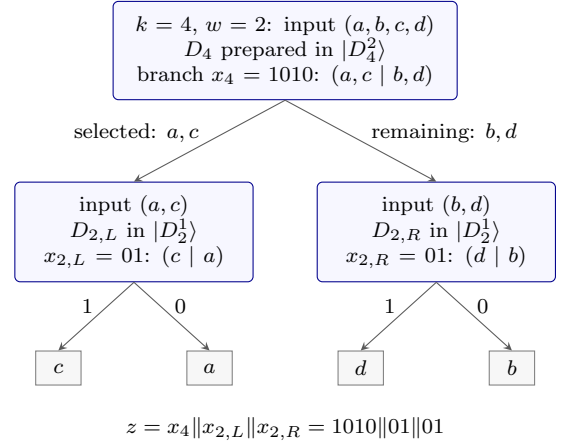

To make the intermediate ordering unambiguous, fix for each block size
$k$ a deterministic procedure that sorts its mask bits in descending
order by pairwise comparisons and swaps.  Section~\ref{subsec:controlled-routing}
defines a sorting network for this purpose and gives its reversible
implementation.  Sorting a mask $x$ produces $1^w0^{k-w}$ and induces a
definite permutation $\tau_x$ of the positions.  Performing the same swaps
on the copies moves precisely the selected subset to the first $w$
positions.  Decode the children on these resulting ordered subblocks and
continue to the leaves.  Equal mask bits need not retain their relative
order: after the subset has been fixed, the children generate every
ordering within that subset.

A \emph{mask collection} contains one mask for every internal node.
Fix the node order $v_1,\ldots,v_{n-1}$ by traversing the tree from the root,
level by level and from left to right.  We represent a collection by the
single binary string $z=x_{v_1}\Vert\cdots\Vert x_{v_{n-1}}$, where
$\Vert$ denotes concatenation.  The mask lengths are fixed by the tree,
so this string has a unique decomposition into its constituent masks.
Writing $N:=\sum_v k_v$, the set of valid encoded collections is
\begin{equation*}
 \mathcal Z_n:=\{z\in\{0,1\}^N:|x_v|=w_v\ \text{for every }v\}.
\end{equation*}
Let $\pi_z$ denote the complete permutation obtained by decoding $z$ from
the root to the leaves.  In Fig.~\ref{fig:mask-tree},
$z=1010\Vert01\Vert01$ first gives $(a,c\mid b,d)$ and then
$(c,a\mid d,b)$.  This demonstrates how local subset choices determine
one global ordering.

\begin{lemma}\label{lem:subset-bijection}
The map $z\mapsto\pi_z$ is a bijection from $\mathcal Z_n$ to $S_n$.
\end{lemma}
\begin{proof}
Temporarily distinguish the input positions by separate symbols.
For any target ordering of a size-$k$ block, the set of its first
$w=\lfloor k/2\rfloor$ symbols uniquely determines the root mask $x$.
The fixed sorting procedure then determines $\tau_x$ and the intermediate
ordering of each child.  Relative to these intermediate orderings,
the target specifies a unique residual permutation in each child.
Induction from a one-element block therefore gives exactly one mask
collection for each target permutation.

The number of choices obeys $k!=\binom{k}{w}w!(k-w)!$ at every node.
Multiplying over the tree cancels the internal factorials and gives
$|\mathcal Z_n|=\prod_v\binom{k_v}{w_v}=n!$.
\end{proof}

For $n=4$, the six root masks and two masks at each child give
$6\cdot2\cdot2=24$ distinct permutations.  At any fixed depth the blocks
are disjoint, so their total size is at most $n$.  Since the tree has at
most $\lceil\log n\rceil$ internal levels, the mask register uses
$N\leq n\lceil\log n\rceil$ qubits.  Equality holds for powers of two.

The uniform superposition of masks at a size-$k$ node is the Dicke state \cite{Bartschi2019Dicke,Bartschi2022Dicke,li2026dicke,chen2026}
\begin{equation}\label{eq:dicke-definition}
 |D_k^w\rangle=\binom{k}{w}^{-1/2}\sum_{|x|=w}|x\rangle.
\end{equation}
Let $P_v$ prepare this state on the mask register $D_v$ and let
$P:=\bigotimes_v P_v$.  All masks are prepared before any copies are
permuted.  Their joint state is
\begin{align}\label{eq:mask-preparation}
 P|0^N\rangle_D
 &=|\Omega_n\rangle_D
   =\bigotimes_v|D_{k_v}^{w_v}\rangle_{D_v}\nonumber\\
 &=\left(\prod_v\binom{k_v}{w_v}^{-1/2}\right)
      \sum_{z\in\mathcal Z_n}|z\rangle_D\nonumber\\
 &=\frac1{\sqrt{n!}}\sum_{z\in\mathcal Z_n}|z\rangle_D.
\end{align}
Together with Lemma~\ref{lem:subset-bijection}, this identifies uniform
mask amplitudes with uniform permutation amplitudes.

\subsection{Coherent preparation of the Dicke masks}
\label{subsec:dicke-preparation}

A sequential sampler prepares Eq.~\eqref{eq:dicke-definition} by keeping
track of the number of ones still to be placed.  Initialize a binary
counter $R$ to $w$ and the $k$ mask qubits to zero.  At position $i$, let
$s=k-i+1$ be the number of remaining positions.  Apply
\begin{equation}\label{eq:dicke-step}
 |r\rangle_R|0\rangle_{D_i}\longmapsto
 |r\rangle_R\left(
 \sqrt{\frac{s-r}{s}}|0\rangle+
 \sqrt{\frac{r}{s}}|1\rangle\right),
\end{equation}
and then reversibly update $R\leftarrow R-D_i$.  The parameter $s$ is
classical, while $r$ can be in superposition.  No measurement is performed.
The branches $r=0$ and $r=s$ apply an exact identity and an exact bit flip,
respectively.

For a fixed output string of weight $w$, the probability numerators from
its ones multiply to $w!$, those from its zeros multiply to $(k-w)!$,
and the denominators multiply to $k!$.  Its probability is consequently
$w!(k-w)!/k!=1/\binom{k}{w}$, independently of the locations of the ones.
All ideal amplitudes are positive.  After the last step $R=0$ on every
branch, so the counter disentangles and the mask is $|D_k^w\rangle$.

To implement Eq.~\eqref{eq:dicke-step}, compute the angle
$\theta_s(r)=2\arcsin\sqrt{r/s}$ into a temporary register, apply the
corresponding $R_y$ rotation using its binary expansion, and uncompute
the angle \emph{before} updating $R$.  Reversible fixed-point
arithmetic~\cite{Bennett1973,Haener2018Arithmetic} and Clifford+T
rotation synthesis~\cite{RossSelinger2016} implement these operations.
Endpoint branches are handled exactly.  In particular, a controlled
rotation is compiled using a half-angle circuit and its exact inverse,
so the inactive branch is exactly the identity.  The compiled sampler
therefore preserves the allowed Hamming weight and returns its counter
and arithmetic workspace to zero.

We detail the preparation error because it determines the precision
required in every mask.  Order the $N$ sampling steps arbitrarily subject
to their dependencies, and denote their ideal and compiled unitaries by
$G_j$ and $\widetilde G_j$.  Include clean temporary qubits implicitly.
It suffices to require
$\|(\widetilde G_j-G_j)|\chi\rangle\|\leq\eta$
for every normalized state $|\chi\rangle$ in the reachable input subspace
of step $j$, where $\eta:=\epsilon/(2N)$.  Let
$|\chi_0\rangle$ be the common initialized state.  Telescoping gives
\begin{align*}
 &(\widetilde G_N\cdots\widetilde G_1-G_N\cdots G_1)|\chi_0\rangle\\
 &\quad=\sum_{j=1}^{N}
 \widetilde G_N\cdots\widetilde G_{j+1}
 (\widetilde G_j-G_j)G_{j-1}\cdots G_1|\chi_0\rangle.
\end{align*}
Each ideal prefix produces a reachable normalized state, and each
compiled suffix is unitary.  The triangle inequality thus yields
\begin{equation}\label{eq:preparation-error}
 \bigl\|\,|\widetilde\Omega_n\rangle-|\Omega_n\rangle\,\bigr\|
 \leq\sum_{j=1}^{N}\eta=N\eta=\epsilon/2,
\end{equation}
where $|\widetilde\Omega_n\rangle:=\widetilde P|0^N\rangle$.

For a rotation, angle error contributes at most
$\|R_y(\widetilde\theta)-R_y(\theta)\|
\leq|\widetilde\theta-\theta|/2$.
Splitting the step error budget between angle evaluation and gate
synthesis requires $b=O(\log(N/\epsilon))=O(\log(n/\epsilon))$ precision
bits, including guard bits.  To evaluate the angle, reduce the arcsine
argument to $\sqrt{\min(r,s-r)/s}\leq1/\sqrt2$ and recover the angle by
symmetry when $r>s/2$.  On this interval, truncating the arcsine series
after $O(b)$ terms gives the required precision.  Horner evaluation uses
$O(b)$ multiplications of $O(b)$-bit numbers, and reversible shift-and-add
multiplication costs $O(b^2)$ Toffoli gates.  Including argument
evaluation, uncomputation, and rotation synthesis gives a conservative
T-count of $O(b^3)$ per sampling step.  Hence
\begin{equation}\label{eq:dicke-cost}
 T(\widetilde P)=O\!\left(n\log n\,\log^3\frac n\epsilon\right).
\end{equation}
Each size-$k$ preparation can use $O(k)$ temporary qubits, allowing all
mask preparations to run in parallel with $O(n\log n)$ total
workspace.\footnote{The arithmetic implementation stores $O(b^2)$
intermediate bits and is used when $b\leq\sqrt k$.  When
$\sqrt k<b\leq k$, use angle lookup with T-count
$O(k\log k+b^2)$ per step.  When $b>k$, use
counter-selected fixed rotations with T-count $O(kb)$ per step.
The latter two implementations use $O(k)$ workspace, and their
T-counts are $O(b^3)$ in their respective ranges.}

\subsection{Reversible sorting and controlled routing}
\label{subsec:controlled-routing}

\begin{definition}[Sorting network]
A sorting network on $k$ wires is a fixed sequence of pairwise
comparison--exchange switches.  A descending switch on wires $i,j$
replaces their keys $(u_i,u_j)$ by
$(\max(u_i,u_j),\min(u_i,u_j))$.  The network sorts every input list in
descending order.  Its wire pairs and their order are independent of
the input values, and its size $f(k)$ is the number of switches.
\end{definition}

\begin{lemma}[Sorting-network size~\cite{AKS1983}]
\label{lem:sorting-network-size}
For every $k\geq2$, there is a sorting network on $k$ numbers with
$f(k)=O(k\log k)$ switches.
\end{lemma}
\begin{proof}
The construction of Ajtai, Koml\'{o}s, and Szemer\'{e}di has
$O(\log k)$ layers, each containing at most $\lfloor k/2\rfloor$
disjoint switches.  Its total size is therefore $O(k\log k)$.
Reversing the comparison convention gives descending order.
\end{proof}

Fix such a network for each block size.  Its fixed wiring also sorts
binary keys, so it can be applied coherently to the computational
branches of a Dicke mask using reversible comparison and
history recording~\cite{Berry2018Sorting}.  A binary quantum comparator acts on two mask
bits $a,b$ and a control bit $t$.  Define $t':=t\oplus(\neg a)b$ and
implement the reversible map
\begin{equation*}
 \mathrm{Comp}:\ |a,b,t\rangle
 \longmapsto |a\oplus t',\ b\oplus t',\ t'\rangle.
\end{equation*}
Thus the record is updated first, and its updated value controls the
two bit flips.  When $t=0$, the first two outputs are
$(\max(a,b),\min(a,b))$, and $t'=1$ precisely for the input $(a,b)=(0,1)$.
This comparator uses one Toffoli, with a negated control implemented by
Pauli $X$ gates, followed by two CNOTs.  Figure~\ref{fig:quantum-comparator}
shows both its three-wire block and its elementary decomposition.

\begin{figure}[tbp]
\centering
\begin{tikzpicture}[x=1cm,y=1cm,
 every node/.style={font=\footnotesize},
 qwire/.style={draw=black,line width=.5pt},
 qgate/.style={draw=blue!55!black,fill=blue!3,line width=.6pt}]
\node[anchor=west] at (0,4.0) {(a) Three-wire comparator};
\foreach \yy/\inlabel in {3.35/a,2.65/b,1.95/t}{
 \draw[qwire] (0.6,\yy)--(7.9,\yy);
 \node[anchor=east] at (0.5,\yy) {$\inlabel$};
}
\draw[qgate] (2.2,1.6) rectangle (3.9,3.7);
\node at (3.05,2.65) {$\mathrm{Comp}$};
\node[fill=white,inner sep=2pt] at (6.15,3.35) {$a\oplus t'$};
\node[fill=white,inner sep=2pt] at (6.15,2.65) {$b\oplus t'$};
\node[fill=white,inner sep=2pt] at (6.15,1.95) {$t'=t\oplus(\neg a)b$};
\node[anchor=west] at (0,1.05) {(b) Elementary implementation};
\foreach \yy/\inlabel in {0.45/a,-0.2/b,-0.85/t}{
 \draw[qwire] (0.6,\yy)--(7.9,\yy);
 \node[anchor=east] at (0.5,\yy) {$\inlabel$};
}
\draw[qwire] (2.0,-0.85)--(2.0,0.45);
\draw[fill=white] (2.0,0.45) circle (2.7pt);
\fill (2.0,-0.2) circle (2pt);
\draw[fill=white] (2.0,-0.85) circle (.09cm);
\draw (1.91,-0.85)--(2.09,-0.85);
\draw (2.0,-0.94)--(2.0,-0.76);
\foreach \xx/\yy in {3.55/0.45,4.75/-0.2}{
 \draw[qwire] (\xx,-0.85)--(\xx,\yy);
 \fill (\xx,-0.85) circle (2pt);
 \draw[fill=white] (\xx,\yy) circle (.09cm);
 \draw (\xx-.09,\yy)--(\xx+.09,\yy);
 \draw (\xx,\yy-.09)--(\xx,\yy+.09);
}
\node[fill=white,inner sep=2pt] at (6.7,0.45) {$a\oplus t'$};
\node[fill=white,inner sep=2pt] at (6.7,-0.2) {$b\oplus t'$};
\node[fill=white,inner sep=2pt] at (6.7,-0.85) {$t'$};
\end{tikzpicture}
\caption{Quantum comparator.  The symbol $\oplus$ denotes addition
modulo two.  The record is updated to $t'=t\oplus(\neg a)b$ before
the two key bits are flipped under its control.  With input $t=0$,
the key outputs are $\max(a,b),\min(a,b)$.  The open control denotes
control on zero.  The inverse comparator reverses the displayed gates.}
\label{fig:quantum-comparator}
\end{figure}

For switch $j$ on a size-$k$ mask, use a fresh bit $t_j$ of the control
register $C$, initially zero.  Immediately after $\mathrm{Comp}_j$,
apply one SWAP to the corresponding pair of copies, controlled by the
recorded value $t_j$.  The next comparator acts on the updated mask.
After the last switch, apply
$\mathrm{Comp}_{f(k)}^\dagger,\ldots,\mathrm{Comp}_1^\dagger$
to the mask and controls only.  These inverse comparators restore the
original mask and return every record bit to zero.  The data swaps are
not included in this inverse pass, as shown explicitly in
Fig.~\ref{fig:clean-partition}.

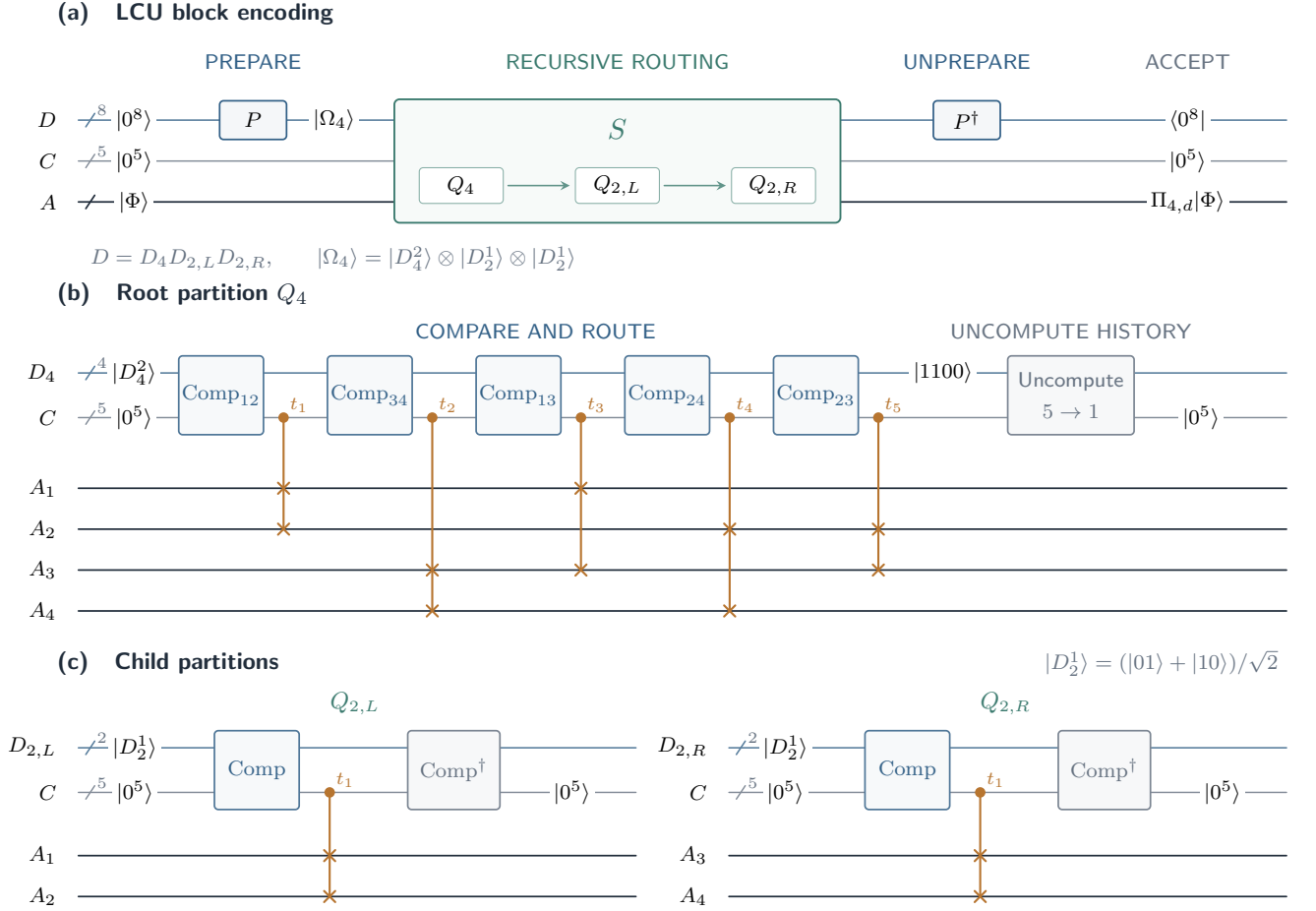
\begin{figure*}[t]
\centering
\begin{tikzpicture}[x=1cm,y=1cm,
  every node/.style={font=\fontsize{8.5}{10}\selectfont},
  line cap=round,line join=round,
  dline/.style={draw=figblue!80,line width=.60pt},
  cline/.style={draw=figgray!80,line width=.55pt},
  aline/.style={draw=figink,line width=.65pt},
  swapline/.style={draw=figamber,line width=.80pt},
  pbox/.style={draw=figblue,fill=figblue!5,rounded corners=1.5pt,
    line width=.65pt,minimum width=.90cm,minimum height=.52cm,inner sep=3pt},
  cmpbox/.style={draw=figblue!85,fill=figblue!4,rounded corners=1.5pt,
    line width=.65pt},
  invbox/.style={draw=figgray!85,fill=figgray!5,rounded corners=1.5pt,
    line width=.65pt},
  tag/.style={font=\sffamily\fontsize{8}{9.5}\selectfont,text=figgray},
  panel/.style={font=\sffamily\bfseries\fontsize{9.5}{11}\selectfont,
    text=figink,anchor=west},
  state/.style={fill=white,inner sep=1.5pt}]
\definecolor{figblue}{HTML}{315F86}
\definecolor{figteal}{HTML}{39796D}
\definecolor{figamber}{HTML}{B77530}
\definecolor{figgray}{HTML}{657180}
\definecolor{figink}{HTML}{202C3A}

\node[panel] at (-.20,1.43) {(a)\quad LCU block encoding};
\node[tag,text=figblue] at (2.55,.79) {PREPARE};
\node[tag,text=figteal] at (7.45,.79) {RECURSIVE ROUTING};
\node[tag,text=figblue] at (12.15,.79) {UNPREPARE};
\node[tag] at (15.12,.79) {ACCEPT};
\foreach \yy/\label/\style in {0/D/dline,-.55/C/cline,-1.10/A/aline}{
  \draw[\style] (.20,\yy)--(16.45,\yy);
  \node[anchor=east] at (.02,\yy) {$\label$};
}
\foreach \yy/\bits/\style in {0/8/dline,-.55/5/cline}{
  \draw[\style] (.30,\yy-.09)--(.42,\yy+.09);
  \node[anchor=south west,font=\sffamily\fontsize{6.5}{7}\selectfont,
    text=figgray,inner sep=.5pt] at (.43,\yy+.04) {$\bits$};
}
\draw[aline] (.30,-1.19)--(.42,-1.01);
\node[state] at (.95,0) {$|0^8\rangle$};
\node[state] at (.95,-.55) {$|0^5\rangle$};
\node[state] at (.95,-1.10) {$|\Phi\rangle$};
\node[pbox] at (2.55,0) {$P$};
\node[state] at (3.64,0) {$|\Omega_4\rangle$};
\draw[draw=figteal,fill=figteal!4,rounded corners=2pt,line width=.70pt]
  (4.45,.28) rectangle (10.45,-1.38);
\node[text=figteal,font=\fontsize{11}{12}\selectfont] at (7.45,-.12) {$S$};
\foreach \xx/\gate in {5.35/Q_4,7.45/{Q_{2,L}},9.55/{Q_{2,R}}}{
  \node[draw=figteal!60,fill=white,rounded corners=1.5pt,
    minimum width=1.13cm,minimum height=.48cm,inner sep=2pt]
    at (\xx,-.88) {$\gate$};
}
\draw[->,>=stealth,draw=figteal!80,line width=.6pt] (5.98,-.88)--(6.81,-.88);
\draw[->,>=stealth,draw=figteal!80,line width=.6pt] (8.08,-.88)--(8.91,-.88);
\node[pbox] at (12.15,0) {$P^\dagger$};
\node[state] at (15.12,0) {$\langle0^8|$};
\node[state] at (15.12,-.55) {$|0^5\rangle$};
\node[state] at (15.12,-1.10) {$\Pi_{4,d}|\Phi\rangle$};
\node[font=\fontsize{8}{9.5}\selectfont,text=figgray,anchor=west]
  at (.25,-1.86)
  {$D=D_4D_{2,L}D_{2,R}$,\qquad
   $|\Omega_4\rangle=|D_4^2\rangle\otimes|D_2^1\rangle\otimes|D_2^1\rangle$};

\begin{scope}[yshift=-3.40cm]
\node[panel] at (-.20,1.06) {(b)\quad Root partition $Q_4$};
\node[tag,text=figblue] at (6.35,.56) {COMPARE AND ROUTE};
\node[tag] at (13.55,.56) {UNCOMPUTE HISTORY};
\foreach \yy/\label/\style in {
  0/{D_4}/dline,-.60/C/cline,
  -1.55/{A_1}/aline,-2.10/{A_2}/aline,
  -2.65/{A_3}/aline,-3.20/{A_4}/aline}{
  \draw[\style] (.20,\yy)--(16.45,\yy);
  \node[anchor=east] at (.02,\yy) {$\label$};
}
\foreach \yy/\bits/\style in {0/4/dline,-.60/5/cline}{
  \draw[\style] (.30,\yy-.09)--(.42,\yy+.09);
  \node[anchor=south west,font=\sffamily\fontsize{6.5}{7}\selectfont,
    text=figgray,inner sep=.5pt] at (.43,\yy+.04) {$\bits$};
}
\node[state] at (.95,0) {$|D_4^2\rangle$};
\node[state] at (.95,-.60) {$|0^5\rangle$};
\foreach \xx/\pair/\jj/\ya/\yb in {
  2.12/12/1/-1.55/-2.10,
  4.12/34/2/-2.65/-3.20,
  6.12/13/3/-1.55/-2.65,
  8.12/24/4/-2.10/-3.20,
  10.12/23/5/-2.10/-2.65}{
  \draw[cmpbox] (\xx-.57,.23) rectangle (\xx+.57,-.83);
  \node[text=figblue] at (\xx,-.30) {$\mathrm{Comp}_{\pair}$};
  \pgfmathsetmacro{\sx}{\xx+.84}
  \draw[swapline] (\sx,-.60)--(\sx,\yb);
  \fill[figamber] (\sx,-.60) circle (2.0pt);
  \node[font=\fontsize{7.5}{9}\selectfont,text=figamber,
    anchor=south west,inner sep=1pt] at (\sx+.045,-.57) {$t_{\jj}$};
  \foreach \yy in {\ya,\yb}{
    \draw[swapline] (\sx-.075,\yy-.075)--(\sx+.075,\yy+.075);
    \draw[swapline] (\sx-.075,\yy+.075)--(\sx+.075,\yy-.075);
  }
}
\node[state,font=\fontsize{8}{9}\selectfont] at (11.83,0) {$|1100\rangle$};
\draw[invbox] (12.70,.23) rectangle (14.40,-.83);
\node[font=\sffamily\fontsize{8}{9.5}\selectfont,text=figgray]
  at (13.55,-.12) {Uncompute};
\node[font=\fontsize{8}{9.5}\selectfont,text=figgray]
  at (13.55,-.52) {$5\rightarrow1$};
\node[state] at (15.30,-.60) {$|0^5\rangle$};
\end{scope}

\begin{scope}[yshift=-8.14cm]
\node[panel] at (-.20,.84) {(c)\quad Child partitions};
\node[anchor=east,font=\fontsize{8}{9.5}\selectfont,text=figgray]
  at (16.45,.84)
  {$|D_2^1\rangle=(|01\rangle+|10\rangle)/\sqrt2$};
\foreach \offset/\side/\first/\second in {0/L/1/2,8.75/R/3/4}{
  \begin{scope}[xshift=\offset cm]
  \node[text=figteal,font=\fontsize{9}{10}\selectfont] at (3.92,.30)
    {$Q_{2,\side}$};
  \foreach \yy/\label/\style in {
    -.30/{D_{2,\side}}/dline,-.90/C/cline,
    -1.75/{A_{\first}}/aline,-2.30/{A_{\second}}/aline}{
    \draw[\style] (.20,\yy)--(7.70,\yy);
    \node[anchor=east] at (.02,\yy) {$\label$};
  }
  \foreach \yy/\bits/\style in {-.30/2/dline,-.90/5/cline}{
    \draw[\style] (.30,\yy-.09)--(.42,\yy+.09);
    \node[anchor=south west,font=\sffamily\fontsize{6.5}{7}\selectfont,
      text=figgray,inner sep=.5pt] at (.43,\yy+.04) {$\bits$};
  }
  \node[state] at (.96,-.30) {$|D_2^1\rangle$};
  \node[state] at (.96,-.90) {$|0^5\rangle$};
  \draw[cmpbox] (2.03,-.07) rectangle (3.17,-1.13);
  \node[text=figblue] at (2.60,-.60) {$\mathrm{Comp}$};
  \draw[swapline] (3.58,-.90)--(3.58,-2.30);
  \fill[figamber] (3.58,-.90) circle (2pt);
  \node[font=\fontsize{7.5}{9}\selectfont,text=figamber,
    anchor=south west,inner sep=1pt] at (3.625,-.87) {$t_1$};
  \foreach \yy in {-1.75,-2.30}{
    \draw[swapline] (3.505,\yy-.075)--(3.655,\yy+.075);
    \draw[swapline] (3.505,\yy+.075)--(3.655,\yy-.075);
  }
  \draw[invbox] (4.63,-.07) rectangle (5.87,-1.13);
  \node[text=figgray] at (5.25,-.60) {$\mathrm{Comp}^\dagger$};
  \node[state] at (6.84,-.90) {$|0^5\rangle$};
  \end{scope}
}
\end{scope}
\end{tikzpicture}
\caption{Permutation-based LCU for four copies.  (a) $P$, $S$, and
$P^\dagger$ implement the accepted operator $\Pi_{4,d}$.
(A slash marks a bundled multi-qubit wire.)
(b,c) The mask inputs are the prepared Dicke states $|D_4^2\rangle$
and $|D_2^1\rangle$.  Routing can entangle masks with data, so no separate
output ket is assigned to a mask wire.  (b) Each recorded $t_j$ controls
the following amber data SWAP.  Inverse comparators $5,4,3,2,1$ clear
$C$ while retaining the data permutation.
(c) Each child compares, swaps, and uncomputes.  The order in (a)
allows reuse of $C$; inverse gates act only on masks and controls.}
\label{fig:clean-partition}
\end{figure*}

To verify the resulting block operation, let $h_k(x)$ denote the
comparison history produced on input mask $x$.  The forward pass gives
\begin{equation*}
 |x\rangle_D|0^{f(k)}\rangle_C|\Phi\rangle_A
 \longmapsto
 |1^w0^{k-w}\rangle_D|h_k(x)\rangle_C
 P_k(\tau_x)|\Phi\rangle_A.
\end{equation*}
The inverse comparator pass acts on the first two factors and yields
the clean partition
\begin{equation}\label{eq:clean-partition}
 Q_k:\ |x\rangle_D|0\rangle_C|\Phi\rangle_A
 \longmapsto |x\rangle_D|0\rangle_C
 P_k(\tau_x)|\Phi\rangle_A.
\end{equation}
Linearity extends this identity to superposed masks and to copies
entangled with other systems.  In particular, the restored mask remains
available for the final LCU unpreparation.

Each network switch induces one forward comparator, one data C-SWAP,
and one inverse comparator.  A qubit C-SWAP decomposes into one Toffoli
and two CNOTs.  Counting three Toffolis per switch and using the standard
seven-T Toffoli decomposition~\cite{Selinger2013} gives
$T(Q_k)\leq21f(k)=O(k\log k)$, with $f(k)$ control qubits for that block.

Apply Eq.~\eqref{eq:clean-partition} at each internal node, with parents
preceding their children.  The resulting unitary $S$ satisfies
\begin{equation}\label{eq:mask-select}
 S:\ |z\rangle_D|0\rangle_C|\Phi\rangle_A
 \longmapsto |z\rangle_D|0\rangle_C
 P_n(\pi_z)|\Phi\rangle_A.
\end{equation}
For $n=2^h$, Lemma~\ref{lem:sorting-network-size} bounds the total number
of network switches by
\begin{align}\label{eq:routing-cost}
 \sum_v f(k_v)
 &=\sum_{\ell=0}^{h-1}2^\ell
       O\!\left(\frac n{2^\ell}\log\frac n{2^\ell}\right)\nonumber\\
 &=O\!\left(n\sum_{\ell=0}^{h-1}(h-\ell)\right)
   =O(n\log^2 n).
\end{align}
For arbitrary $n$, the tree still has $O(\log n)$ levels, each of total
block size at most $n$, giving the same bound.  Processing blocks
successively lets the control register be reused; its peak size is
$\max_v f(k_v)=O(n\log n)$.

\subsection{Implementing the symmetric projector}
\label{subsec:projection-synthesis}

We now derive the complete LCU operation~\cite{Gilyen2019QSVT}.
Suppress the control and arithmetic registers that start and end in
zero, and write $S$ for its restriction to this clean-workspace
subspace.  Equation~\eqref{eq:mask-select} says that, on valid masks,
$S$ acts as $\sum_{z\in\mathcal Z_n}|z\rangle\!\langle z|_D
\otimes P_n(\pi_z)$, with a unitary extension on the remaining basis
states.  Define $U_{n,d}:=(P^\dagger\otimes I_A)S(P\otimes I_A)$.
For an arbitrary input vector $|\Phi\rangle_A$, the preparation in
Eq.~\eqref{eq:mask-preparation} and the routing give
\begin{align*}
 |0^N\rangle_D|\Phi\rangle_A
 &\xrightarrow{P\otimes I_A}
 \frac1{\sqrt{n!}}\sum_{z\in\mathcal Z_n}|z\rangle_D|\Phi\rangle_A\\
 &\xrightarrow{S}
 \frac1{\sqrt{n!}}\sum_{z\in\mathcal Z_n}
 |z\rangle_D P_n(\pi_z)|\Phi\rangle_A\\
 &\xrightarrow{P^\dagger\otimes I_A}
 \frac1{\sqrt{n!}}\sum_{z\in\mathcal Z_n}
 P^\dagger|z\rangle_D P_n(\pi_z)|\Phi\rangle_A.
\end{align*}
For every valid $z$, preparation implies
$\langle0^N|P^\dagger|z\rangle=1/\sqrt{n!}$.
Projecting the mask onto zero therefore gives
\begin{align}\label{eq:projector-block}
 &(\langle0^N|_D\otimes I_A)U_{n,d}
       (|0^N\rangle_D\otimes|\Phi\rangle_A)\nonumber\\
 &\quad=\frac1{\sqrt{n!}}\sum_{z\in\mathcal Z_n}
 \langle0^N|P^\dagger|z\rangle
 P_n(\pi_z)|\Phi\rangle_A\nonumber\\
 &\quad=\frac1{n!}\sum_{z\in\mathcal Z_n}
 P_n(\pi_z)|\Phi\rangle_A
 =\Pi_{n,d}|\Phi\rangle_A.
\end{align}
The last equality uses Lemma~\ref{lem:subset-bijection}.  Since
$|\Phi\rangle_A$ is arbitrary, Eq.~\eqref{eq:projector-block} proves that
the encoded operator is $\Pi_{n,d}$, with normalization one.
Unpreparation does not deterministically reset the masks after routing;
the all-zero measurement outcome selects this block.  On input
$\rho_\psi^{\otimes n}$, the unnormalized accepted state is
$\Pi_{n,d}\rho_\psi^{\otimes n}\Pi_{n,d}$, with probability $p_n$.
Discarding all but one copy then implements the purification map of
Sec.~\ref{subsec:symmetric-projection-protocol}.

For the error analysis, use the compiled preparation $\widetilde P$
and its exact circuit inverse, so
$\widetilde U_{n,d}:=(\widetilde P^\dagger\otimes I_A)
S(\widetilde P\otimes I_A)$.
The ideal and compiled encoded blocks can be written as
\begin{align*}
 \Pi_{n,d}&=(\langle\Omega_n|\otimes I_A)
        S(|\Omega_n\rangle\otimes I_A),\\
 \widetilde\Pi_{n,d}&=(\langle\widetilde\Omega_n|\otimes I_A)
        S(|\widetilde\Omega_n\rangle\otimes I_A).
\end{align*}
Set $|e\rangle:=|\widetilde\Omega_n\rangle-|\Omega_n\rangle$.
Adding and subtracting the term with compiled bra and ideal ket gives
\begin{align*}
 \widetilde\Pi_{n,d}-\Pi_{n,d}
 &=(\langle e|\otimes I_A)S(|\Omega_n\rangle\otimes I_A)\\
 &\quad+(\langle\widetilde\Omega_n|\otimes I_A)
             S(|e\rangle\otimes I_A).
\end{align*}
For any vector $|v\rangle$, the map $|v\rangle\otimes I_A$ has operator
norm $\|v\|$.  Both preparation states have norm one, and $\|S\|=1$.
Submultiplicativity and Eq.~\eqref{eq:preparation-error} therefore imply
\begin{equation}\label{eq:block-error}
 \|\widetilde\Pi_{n,d}-\Pi_{n,d}\|
 \leq\|e\|+\|e\|\leq\epsilon.
\end{equation}
This is an operator-norm guarantee for the accepted Kraus operator.
It is independent of the local dimension and introduces no additional
approximation in the sorting or data swaps.

\subsection{Total resources}
\label{subsec:total-resources}

\upperbound*
\begin{proof}
Encode each input copy in $q:=\lceil\log d\rceil$ qubits.  Swapping two
copies under one control then uses $q$ qubit C-SWAPs; mask preparation
and comparison are unchanged.  Including all $m$ auxiliary qubits,
Eq.~\eqref{eq:block-error} gives
\begin{equation*}
 \bigl\|(\langle0^m|\otimes I_A)\widetilde U_{n,d}
       (|0^m\rangle\otimes I_A)-\Pi_{n,d}\bigr\|\leq\epsilon.
\end{equation*}
The two preparation passes have the cost in Eq.~\eqref{eq:dicke-cost}.
For routing, each network switch has T-count at most $7(q+2)$:
two Toffolis for the forward and inverse comparators and $q$ for the
controlled data swaps.  Equation~\eqref{eq:routing-cost} thus gives
\begin{align*}
 T(\widetilde U_{n,d})
 &=2T(\widetilde P)+T(S)\\
 &=O\!\left(n\log n\,\log^3\frac n\epsilon
       +n\log^2n\,\log d\right)
 =\widetilde O(n).
\end{align*}
The masks, the reusable control register, and the parallel preparation
workspace each occupy $O(n\log n)$ qubits.  Their sum is
$m=O(n\log n)$, in addition to the $n\lceil\log d\rceil$ data qubits.
\end{proof}

\section{Lower bounds}\label{sec4}

We prove a linear lower bound for the operator-synthesis task in
Sec.~\ref{sec3}: approximating the symmetric projector in operator norm.
The proof uses the stabilizer-nullity framework of
Beverland et al.~\cite{Beverland2020Nullity} and its unitary formulation
by Jiang and Wang~\cite{JiangWang2023Nullity}.  We first explain how
nullity bounds the cost of a postselected circuit.  We then show that
the qubit symmetrizer has only four Pauli symmetries, separated from
all other Pauli pairs by a constant operator-norm gap.  This gap makes
the lower bound robust to approximation.  Finally, restricting each
input copy to two computational-basis levels extends the result to
every local dimension $d\geq2$.

Throughout this section, the T-count $T(V)$ counts both T and T$^\dagger$ gates
in the implemented Clifford+T circuit $V$.  Clifford operations and
the preparation of stabilizer states are free.  We use $\Pi_{n,d}$ for the
symmetric projector throughout; in the qubit argument, this becomes
$\Pi_{n,2}$.
The norm $\|\cdot\|$ is the operator norm; $\|\cdot\|_{\mathrm F}$
denotes the Frobenius norm.

\subsection{Stabilizer nullity and postselection}
\label{subsec:nullity-postselection}

Let $\mathcal Q_r:=\{I,X,Y,Z\}^{\otimes r}$ be the set of $r$-qubit
Pauli strings, and let $\widehat{\mathcal P}_r$ be the Pauli group
obtained by adjoining the phases $\{1,-1,i,-i\}$.  For a nonzero pure
state $|\psi\rangle$ on $r$ qubits, define
\begin{align*}
 \operatorname{Stab}(\psi)
 &:=\{Q\in\widehat{\mathcal P}_r:Q|\psi\rangle=|\psi\rangle\},\\
 \nu_s(\psi)&:=r-\log_2|\operatorname{Stab}(\psi)|.
\end{align*}
Its stabilizer group is abelian and contains no nonidentity scalar.
Consequently, it has $2^s$ elements for some $0\leq s\leq r$;
$\nu_s$ counts the number of qubits not fixed by independent Pauli
stabilizers.  In particular, a pure stabilizer state has nullity zero.

For an $r$-qubit unitary $W$, the unitary stabilizer nullity counts
the Pauli strings that remain Pauli under conjugation:
\begin{align*}
 s_u(W)&:=\bigl|\{Q\in\mathcal Q_r:
       WQW^\dagger\in\widehat{\mathcal P}_r\}\bigr|,\\
 \nu_u(W)&:=2r-\log_2 s_u(W).
\end{align*}
Jiang and Wang~\cite{JiangWang2023Nullity} show that
\begin{gather*}
 \nu_u(W)=\nu_s\bigl((W\otimes I)|\Phi_r\rangle\bigr)
 \leq T(W),\\
 |\Phi_r\rangle:=2^{-r/2}\sum_{x\in\{0,1\}^r}|x\rangle|x\rangle.
\end{gather*}
Thus unitary nullity can be evaluated through the state's Pauli
stabilizers after applying the unitary to one half of Bell pairs.
Our target $\Pi_{n,2}$ is a projector, so we use this Choi-state viewpoint
on the accepted block of its unitary implementation.  The following
fact justifies the required postselection.

\begin{lemma}[Nullity bound with postselection]
\label{lem:nullity-postselection}
Suppose a Clifford+T circuit starts in a pure stabilizer state and
uses $t$ gates from $\{T,T^\dagger\}$.  If computational-basis
postselection with nonzero probability leaves a pure state
$|\phi\rangle$ on the unmeasured qubits, then $t\geq\nu_s(\phi)$.
\end{lemma}
\begin{proof}
Clifford conjugation is a bijection of the Pauli group, so it preserves
the size of the stabilizer group and hence $\nu_s$.  Let
$S=\operatorname{Stab}(\psi)$ immediately before a T or
T$^\dagger$ gate on qubit $j$.  The subgroup
$S_0:=\{Q\in S:[Q,Z_j]=0\}$ has index at most two in $S$: the
commutation sign with $Z_j$ is a homomorphism from $S$ to
$\{1,-1\}$.  Every member of $S_0$ commutes with the gate and
remains a stabilizer afterwards.  Thus one T or T$^\dagger$
can increase $\nu_s$ by at most one.

Now postselect qubit $j$ onto $|0\rangle$, with nonzero remaining
state $|\phi\rangle$.  Each element of $S_0$ restricts to a Pauli
stabilizer of $|\phi\rangle$.  The restriction homomorphism has
kernel contained in $\{I,Z_j\}$.  If the kernel is trivial, its
image has size $|S_0|\geq|S|/2$.  If the kernel has size two, then
$Z_j\in S$.  Since $S$ is abelian, this implies $S_0=S$, and the
image again has size $|S|/2$.  In either case,
\begin{align*}
 |\operatorname{Stab}(\phi)|&\geq|S|/2,\\
 \nu_s(\phi)&\leq(r-1)-\log_2(|S|/2)=\nu_s(\psi).
\end{align*}
Postselection on $|1\rangle$ reduces to this case by a Clifford
$X_j$, and repeated application handles several measured qubits.
Starting from nullity zero and applying these bounds gate by gate
proves the claim.  Tensoring with a pure stabilizer state also
preserves nullity: a Clifford maps that factor to $|0\rangle^{\otimes w}$,
whose $2^w$ stabilizers compensate for its $w$ additional qubits.
\end{proof}

\subsection{Pauli symmetries of the qubit symmetrizer}
\label{subsec:pauli-symmetrizer}

A Pauli string is \emph{collective} if it has the form
$\sigma^{\otimes n}$ for $\sigma\in\{I,X,Y,Z\}$.  Such a string
acts identically on every copy and commutes with every permutation,
and hence with $\Pi_{n,2}$.  The next lemma shows that all other Pauli
pairs have a finite gap from being a two-sided symmetry of $\Pi_{n,2}$.
Set $\Delta_n:=\sqrt{1-1/n}$.

\begin{lemma}[Pauli symmetry gap]
\label{lem:pauli-symmetry-gap}
Let $n\geq2$, $Q,R\in\mathcal Q_n$, and $|\omega|=1$.
Unless $Q=R=\sigma^{\otimes n}$ for some
$\sigma\in\{I,X,Y,Z\}$, one has
$\|Q\Pi_{n,2} R^{\mathsf T}-\omega\Pi_{n,2}\|\geq\Delta_n$.
\end{lemma}
\begin{proof}
Suppose first that $R$ is not collective.  Up to an overall phase,
write $R^{\mathsf T}=r_1\otimes\cdots\otimes r_n$ with
$r_j\in\{I,X,Y,Z\}$.  Since the factors are not all equal, there
is an $h\in\{X,Y,Z\}$ that commutes with some $r_j$ and
anticommutes with others.  Indeed, the three commutation signs with
$X,Y,Z$ distinguish the four one-qubit Pauli strings.

Choose normalized eigenvectors $|0_h\rangle,|1_h\rangle$ of $h$
with eigenvalues $+1,-1$, respectively, and put
$|x_h\rangle:=|0_h\rangle^{\otimes n}$.  This state is symmetric,
so $\Pi_{n,2}|x_h\rangle=|x_h\rangle$.  The state
$R^{\mathsf T}|x_h\rangle$ is, up to phase, a product of these
eigenvectors with $k$ factors equal to $|1_h\rangle$, where
$1\leq k\leq n-1$.  Write $N_k:=\binom nk$ and define the
Dicke state in this basis by
\begin{equation*}
 |D_n^{k,h}\rangle:=\frac1{\sqrt{N_k}}
   \sum_{\substack{z\in\{0,1\}^n\\|z|=k}}
   |z_1\rangle_h\otimes\cdots\otimes|z_n\rangle_h.
\end{equation*}
Permutation averaging therefore gives, for some real $\theta$,
\begin{equation*}
 \Pi_{n,2} R^{\mathsf T}|x_h\rangle
   =\frac{e^{i\theta}}{\sqrt{N_k}}|D_n^{k,h}\rangle,
 \qquad
 \|\Pi_{n,2} R^{\mathsf T}|x_h\rangle\|^2=\frac1{N_k}.
\end{equation*}
Likewise, $Q^\dagger|x_h\rangle$ is an $h$-basis product state
up to phase.  Its overlap with $|D_n^{k,h}\rangle$ has magnitude
either zero or $1/\sqrt{N_k}$.  Hence
$|\langle x_h|Q\Pi_{n,2} R^{\mathsf T}|x_h\rangle|\leq1/N_k$.
Expanding the squared norm gives
\begin{align*}
 &\|(Q\Pi_{n,2} R^{\mathsf T}-\omega\Pi_{n,2})|x_h\rangle\|^2\\
 &\quad=1+\frac1{N_k}
  -2\operatorname{Re}\!\left[
    \overline\omega\langle x_h|Q\Pi_{n,2} R^{\mathsf T}|x_h\rangle
   \right]\\
 &\quad\geq1-\frac1{N_k}\geq1-\frac1n.
\end{align*}
The final inequality uses $\binom nk\geq n$ for
$1\leq k\leq n-1$.  This proves the claimed operator-norm bound
whenever $R$ is noncollective.  If $Q$ is noncollective, take the
adjoint, which preserves the norm, and apply the same argument with
$Q^\dagger$ as the right Pauli factor.

It remains to consider distinct collective strings $Q$ and $R$.
They commute with $\Pi_{n,2}$, and their product, including the transpose,
is a phase times $p^{\otimes n}$ for a nonidentity
$p\in\{X,Y,Z\}$.  Absorb this phase into $\omega$.  On the
symmetric subspace, $p^{\otimes n}$ has eigenvalue $+1$ on
$|0_p\rangle^{\otimes n}$ and eigenvalue $-1$ on
$|D_n^{1,p}\rangle$.  Consequently, its distance there from any
scalar phase obeys
\begin{align*}
 \|p^{\otimes n}\Pi_{n,2}-\omega\Pi_{n,2}\|
 &\geq\max\{|1-\omega|,|-1-\omega|\}\\
 &\geq\sqrt2\geq\Delta_n.
\end{align*}
This completes all cases.
\end{proof}

For later use, define the vectorization of an operator $M$ by
$\operatorname{vec}(M):=\sum_{x,y}M_{y,x}|y\rangle_A|x\rangle_{A'}$.
Here $A'$ is a reference system with the same dimension as the data
system $A$.  The elementary identity
\begin{equation*}
 (Q\otimes R)\operatorname{vec}(M)
   =\operatorname{vec}(QMR^{\mathsf T})
\end{equation*}
turns two-sided Pauli symmetries into stabilizers of a pure state.
For the exact qubit projector, Lemma~\ref{lem:pauli-symmetry-gap}
implies that only the four matched collective Pauli labels can
stabilize $\operatorname{vec}(\Pi_{n,2})$.  All four occur, with the
stabilizing phases given by
\begin{align*}
 \operatorname{Stab}(\operatorname{vec}(\Pi_{n,2}))
 =\bigl\{&I^{\otimes 2n},\ X^{\otimes 2n},\\
         &Z^{\otimes 2n},\ (-1)^nY^{\otimes 2n}\bigr\}.
\end{align*}
Thus its nullity is exactly $2n-2$.  To treat approximation, we
next use the quantitative gap to bound the stabilizers of the
\emph{implemented} block directly.

\subsection{Robust lower bound for block encodings}
\label{subsec:robust-lower-bound}

For a unitary $U$ acting on the data and $m$ signal qubits, denote
its accepted block by
\begin{equation*}
 K_U:=(\langle0^m|\otimes I_A)U(|0^m\rangle\otimes I_A).
\end{equation*}
An $(\alpha,m,\epsilon)$ block encoding of $\Pi_{n,d}$ satisfies
\begin{equation}\label{eq:lower-block-condition}
 \|\Pi_{n,d}-\alpha K_U\|\leq\epsilon,
 \qquad 0<\alpha<\infty.
\end{equation}
For qudits, use the standard binary encoding of
Sec.~\ref{subsec:projection-synthesis}; the norm condition is imposed
on the valid encoded data subspace.  The signal qubits include all
qubits whose all-zero outcome selects the block.  In addition, an
implementation $V$ may use any number of work qubits in a pure
stabilizer state $|s\rangle$, provided that they are returned
unchanged and uncorrelated:
\begin{equation*}
 V(|\psi\rangle\otimes|s\rangle)
   =U|\psi\rangle\otimes|s\rangle
 \quad\text{for every input }|\psi\rangle.
\end{equation*}
The approximation is in the encoded block; $V$ denotes the actual
finite Clifford+T circuit.

\lowerbound*
\begin{proof}
First let $d=2$, set $M:=\alpha K_U$, and let $t:=T(V)$.
Prepare the stabilizer state
$|0^m\rangle\otimes|\Phi_n\rangle_{AA'}\otimes|s\rangle$,
where each data qubit is Bell-paired with its reference qubit.
Apply $V$ to the signal, data, and work qubits, leaving $A'$ idle,
and postselect the signal onto zero.  The unnormalized remaining
state is
\begin{equation*}
 2^{-n/2}\operatorname{vec}(K_U)_{AA'}\otimes|s\rangle.
\end{equation*}
Since $\|\Pi_{n,2}\|=1$ and $\epsilon<1$, the approximation condition
ensures $M\neq0$, so this outcome has nonzero probability.
Its normalized state on $AA'$ is
\begin{equation*}
 |\Xi_M\rangle:=\frac{\operatorname{vec}(M)}{\|M\|_{\mathrm F}}
   =\frac{\operatorname{vec}(K_U)}{\|K_U\|_{\mathrm F}}.
\end{equation*}
The clean stabilizer work factor does not affect nullity.
Lemma~\ref{lem:nullity-postselection} therefore gives
\begin{equation}\label{eq:choi-block-nullity}
 t\geq\nu_s(\Xi_M).
\end{equation}

Take any Pauli stabilizer of $|\Xi_M\rangle$ and write it as
$\gamma Q_A\otimes R_{A'}$, where $Q,R\in\mathcal Q_n$ and
$\gamma\in\{1,-1,i,-i\}$.  Vectorization implies
$QMR^{\mathsf T}=\omega M$ with $\omega=\gamma^{-1}$.
Using this exact relation for the implemented block, together with
Eq.~\eqref{eq:lower-block-condition}, yields
\begin{align*}
 &\|Q\Pi_{n,2} R^{\mathsf T}-\omega\Pi_{n,2}\|\\
 &\quad=\|Q(\Pi_{n,2}-M)R^{\mathsf T}+\omega(M-\Pi_{n,2})\|\\
 &\quad\leq\|Q(\Pi_{n,2}-M)R^{\mathsf T}\|+\|M-\Pi_{n,2}\|\\
 &\quad=2\|M-\Pi_{n,2}\|\leq2\epsilon<\Delta_n.
\end{align*}
Lemma~\ref{lem:pauli-symmetry-gap} then forces
$Q=R=\sigma^{\otimes n}$ for some $\sigma\in\{I,X,Y,Z\}$.
A nonzero state has at most one stabilizing phase for each Pauli
string.  There are only four possible strings, and hence
\begin{equation*}
 |\operatorname{Stab}(\Xi_M)|\leq4,
 \qquad \nu_s(\Xi_M)\geq2n-2.
\end{equation*}
Combining this with Eq.~\eqref{eq:choi-block-nullity} proves the
qubit bound.  This step uses operator-norm error before normalizing
the Choi vector; it does not require a continuity bound for nullity
or a lower bound on the postselection probability.

For general $d\geq2$, each copy occupies $q:=\lceil\log d\rceil$
qubits.  Define the isometry $J$ by
$J|b\rangle:=|0^{q-1}b\rangle$ for $b\in\{0,1\}$, and let
$J_n:=J^{\otimes n}$.  Both encoded levels are valid for every
$d\geq2$.  Permuting the copies preserves this subspace and
commutes with applying the same embedding to each copy.  Averaging
over permutations therefore gives
\begin{equation*}
 \Pi_{n,d}J_n=J_n\Pi_{n,2},
 \qquad J_n^\dagger\Pi_{n,d}J_n=\Pi_{n,2}.
\end{equation*}
Regard the $n(q-1)$ padding qubits as additional signal qubits,
initialized and postselected in zero.  The same circuit, with only
a relabeling of wires, has accepted block
$K_2:=J_n^\dagger K_UJ_n$ on $n$ qubits and satisfies
\begin{align*}
 \|\Pi_{n,2}-\alpha K_2\|
 &=\|J_n^\dagger(\Pi_{n,d}-\alpha K_U)J_n\|\leq\epsilon.
\end{align*}
The restriction does not increase the T-count.  Applying the qubit result to
this $(\alpha,m+n(q-1),\epsilon)$ block encoding proves the
claim for all $d\geq2$.
\end{proof}

For every fixed $0\leq\epsilon_0<1/2$, the theorem applies uniformly
to $0\leq\epsilon\leq\epsilon_0$ whenever
$n>1/(1-4\epsilon_0^2)$.  Indeed, this condition is equivalent to
$2\epsilon_0<\sqrt{1-1/n}$.  Thus the lower bound remains
$2n-2=\Omega(n)$ at any fixed error strictly below $1/2$.
For fixed $0<\epsilon<1/2$, combining it with
Theorem~\ref{thm:symmetric-projection-synthesis} places the minimum
T-count for normalization-one block encodings between
$\Omega(n)$ and $\widetilde O(n)$, where the upper bound suppresses
logarithmic factors in $n$, $d$, and $1/\epsilon$.
The dependence on the number of input copies is therefore optimal
up to logarithmic factors.

This conclusion concerns implementing the accepted symmetric projector with uniform operator-norm accuracy.  A lower bound for
arbitrary CPTN purification maps specified only by their action on
noisy product inputs, or by their probability--fidelity trade-off,
would require an additional reduction to this operator-synthesis task.

\section{Summary}\label{sec5}

We develop a Clifford+T synthesis framework for
symmetric projector purification protocols under depolarizing noise. Dicke
masks arranged on a balanced binary tree specify uniform permutation
instructions, and reversible sorting networks coherently route the input copies. Combining these components with LCU yields a normalization-one block encoding of the symmetric projector.
For $n\geq2$ copies with local dimension $d\geq2$ and any
$0<\epsilon<1/2$, the circuit achieves operator-norm error at most
$\epsilon$ with T-count $\widetilde O(n)$ and $O(n\log n)$ auxiliary
qubits. The suppressed factors are logarithmic in $n$, $d$, and
$1/\epsilon$, giving only logarithmic dependence on the local dimension.

Using the stabilizer nullity of the postselected state together with a quantitative Pauli symmetry gap, we proved that any block encoding
with $\epsilon<\tfrac12\sqrt{1-1/n}$ has a T-count of at least $2n-2$.
This bound holds for all $d\geq2$, arbitrary positive
normalization, and any number of clean stabilizer work qubits.
Taken together, the upper and lower bounds establish near-optimal scaling with $n$ for synthesizing the symmetric projector at any fixed
$0<\epsilon<1/2$, where the lower bound applies for sufficiently large $n$.

Future work may reduce the logarithmic overhead and auxiliary space
through improved Dicke-state preparation and permutation routing.
Sharper lower bounds in $d$ and $1/\epsilon$ would clarify the
remaining parameter dependence.  Extending the lower bound to
arbitrary purification maps specified by a probability--fidelity
trade-off also remains open: the present optimality result concerns
uniform operator-norm approximation of the symmetric projector.
Such an extension would connect operator synthesis to the resources
required by the purification task itself.

\section{Acknowledgment}
The authors thank Longcheng Li, Jiaqing Jiang, Zhiyu Xia and Jialin Zhang for helpful discussions. 

\FloatBarrier
\sloppy
\bibliographystyle{apsrev4-1} 
\bibliography{ref}

\end{document}